\documentclass[12pt,a4paper]{amsart}
\usepackage[bottom]{footmisc}
\usepackage{mathrsfs}
\usepackage{amsmath}
\usepackage{amsthm}
\usepackage{amsfonts}
\usepackage{latexsym}
\usepackage{graphicx}
\usepackage{hyperref}
\usepackage{amssymb}
\usepackage{subfigure}
\usepackage{epsf}
\usepackage{float}
\usepackage{hyperref}
\usepackage{fancyhdr}
\allowdisplaybreaks

\begin{document}
\newcommand{\beq}{\begin{equation}}
\newcommand{\eneq}{\end{equation}}
\newtheorem{thm}{Theorem}[section]
\newtheorem{coro}[thm]{Corollary}
\newtheorem{lem}[thm]{Lemma}
\newtheorem{prop}[thm]{Proposition}
\newtheorem{defi}[thm]{Definition}
\newtheorem{rem}[thm]{Remark}
\newtheorem{cl}[thm]{Claim}
\title{A new approach for the strong unique continuation of electromagnetic
Schr\"{o}dinger operator with complex-valued coefficient}
\author{Xiaojun Lu$^{1}$\ \ \ \ \ \ \ Xiaofen Lv$^2$}
\thanks{Corresponding author: Xiaojun Lu, Department of Mathematics \& Jiangsu Key Laboratory of
Engineering Mechanics, Southeast University, 210096, Nanjing, China}
\date{}
\maketitle
\pagestyle{fancy}                   % ÉèÖÃҳü
\lhead{X. Lu and X. Lv} \rhead{Unique continuation for
electromagnetic
Schr\"{o}dinger operator} %\rhead{\small\leftmark}
\begin{center}
1. Department of Mathematics \& Jiangsu Key Laboratory of
Engineering Mechanics, Southeast University, 210096, Nanjing, China\\
2. Jiangsu Testing Center for Quality of Construction Engineering
Co., Ltd, 210028, Nanjing, China
%2. Department of Sociology, School of Public Administration, Hohai
%University, 211189, Nanjing, China
\end{center}
\begin{abstract}
This paper mainly addresses the strong unique continuation property
for the electromagnetic Schr\"{o}dinger operator with complex-valued
coefficients. Appropriate multipliers with physical backgrounds have
been introduced to prove a priori estimates. Moreover, its
application in an exact controllability problem has been shown, in
which case, the boundary value determines the interior value
completely.
\end{abstract}
\renewcommand{\abstractname}{R\'{e}sum\'{e}}
\begin{abstract}
Dans cet article, on consid\`{e}re essentiellement la
propi\'{e}t\'{e} fortement unique pour l'op\'{e}rateur
\'{e}lectromagn\'{e}tique de Schr\"{o}dinger avec les
co\'{e}fficients de complexe. La m\'{e}thode de multiplicateur a
\'{e}t\'{e} introduite pour d\'{e}montrer les estimations \`{a}
priori. En plus, cette th\'{e}orie est appliqu\'{e}e dans le
probl\`{e}me de contr\^{o}labilit\'{e} exacte, o\`{u} la valeur sur
la fronti\`{e}re d\'{e}terminera la valeur int\'{e}rieure
compl\`{e}tement.
\end{abstract}

\section{Introduction}
Nowadays, quantum studies, especially multiphoton entanglement and
interferometry, are attracting many scientists' attention, either
theoretically or practically\cite{P}. A few world-famous high-tech
companies, such as Apple, Microsoft, etc. are developing new
generation of high-performance computers based on the quantum
mechanics
phenomena.\\

In our paper, we discuss an important complex-valued operator in
this research field. Let ${\bf A}(x)$ be the vector potential of the
magnetic field ${\bf B}$, that is, ${\bf B}=\nabla\times{\bf A}$.
Clearly, $ \nabla\cdot{\bf B}={\rm div}\ {\rm rot}{\bf A}=0.$ From
one of Maxwell's equations($\mu$ is magnetic permeability)
$\nabla\times{\bf E}=-\mu{\partial{\bf B}}/{\partial t}=0,$ we
deduce that ${\bf E}=-\nabla\phi,$ where the scalar $\phi$
represents the electric potential. We choose an appropriate
Lagrangian for the non-relativistic  charged particle in the
electromagnetic field ($q$ is the electric charge of the particle,
and ${\bf v}$ is its velocity, $m$ is mass), $\mathscr{L}={m{\bf
v}^2}/2-q\phi+q{\bf v}\cdot{\bf A}.$ Particularly, the canonical
momentum is specified by the vector ${\bf p}=\nabla_{\bf
v}\mathscr{L}=m{\bf v}+q{\bf A}.$ Next, we define the classical
Hamiltonian by Legendre transform, $H\triangleq{\bf p}\cdot{\bf
v}-\mathscr{L}={({\bf p}-q{\bf A})^2}/{(2m)}+q\phi.$ In quantum
mechanics, when ${\bf p}$ is replaced by $-i\hbar\nabla$,($\hbar$ is
the Planck constant), we have the following operator
\begin{equation}P\triangleq{(i\hbar\nabla+q{\bf A})^2}/{(2m)}+q\phi: \mathcal{H}\to\mathcal{H}^\ast,\end{equation}
where $\mathcal{H}$ and $\mathcal{H}^\ast$ are corresponding
function spaces. Lots of literature is devoted to the research of
this kind of
operator\cite{LF,MG,Koch,MT}.\\

Let $\Omega\subset\mathbb{R}^N$ be an open, connected and bounded
domain. From the structure of operator $P$, we define the
corresponding simplified operators
\begin{equation}\mathscr{H}_{\bf A}\triangleq i\nabla+{\bf A}(x):
L^2(\Omega)\to (L^2(\Omega))^N,\end{equation}
\begin{equation}\mathscr{H}_{\bf A}^2\triangleq(i\nabla+{\bf A}(x))^2:
L^2(\Omega)\to L^2(\Omega),\end{equation} where ${\bf A}\in
C^1(\overline{\Omega})$ is a real-valued potential vector. The
corresponding derivative of the magnetic potential ${\bf A}$ is as
follows,
$$D{\bf A}=\left(\begin{array}{cccc}\nabla_{1}a_1&\nabla_{1}a_2,&\cdots&\nabla_{1}a_N\\
\nabla_{2}a_1&\nabla_{2}a_2&\cdots&\nabla_{2}a_N\\
\vdots&\vdots&\cdots&\vdots\\
\nabla_{N}a_1&\nabla_{N}a_2&\cdots&\nabla_{N}a_N\\
\end{array}\right),$$
where $$\nabla_{i}a_j\triangleq{\partial a_j}/{\partial x_i},
i,j=1,\cdots,N.$$ In addition, one defines the following $N\times N$
anti-symmetric matrix $\Xi_{\bf A}$ given by
$$\Xi_{\bf A}\triangleq (D{\bf A})^T-D{\bf A}^T=\left(\begin{array}{cccc}\xi_{11}&\xi_{12},&\cdots&\xi_{1N}\\
\xi_{21}&\xi_{22}&\cdots&\xi_{2N}\\
\vdots&\vdots&\cdots&\vdots\\
\xi_{N1}&\xi_{N2}&\cdots&\xi_{NN}\\
\end{array}\right) $$
with
$$\xi_{jk}\triangleq \nabla_j a_k-\nabla_k a_j,\ \ k,j=1,\cdots,N.$$
%$$B=(D{\bf A})-(D{\bf A})^t,\ \ \ B_{kj}=(\frac{\partial A_k}{\partial x_j}-\frac{\partial A_j}{\partial x_k}),\ \ \ k,j=1,\cdots,N.$$
In quantum mechanics, $\Xi_{\bf A}\equiv0$ stands for the case
without magnetic field, i.e. $${\bf B}={\rm rot}{\bf A}=0.$$ Once
the magnetic field exists, then $\Xi_{\bf A}\neq0$. Consequently,
$\Xi_{\bf A}$ serves as a test matrix for the magnetic field.
Interested readers can refer to \cite{LF1,LU1,NM} for more details
concerned with the vector operator $\mathscr{H}_{\bf A}$ and
self-adjoint operator $\mathscr{H}_{\bf A}^2$. In such a manner, (1)
is simplified as
\begin{equation}
\mathscr{H}_{\bf A}^2-\phi(x): L^2(\Omega)\to L^2(\Omega),
\end{equation}
where the complex-valued function $\phi\in L^\infty(\Omega)$. In
this paper, we focus on the strong unique continuation
property(SUCP) for the electromagnetic Schr\"{o}dinger operator (4).
In the following, we introduce a few important definitions.
\begin{defi} A function $u\in L^2_{loc}(\Omega)$ is said to
vanish of infinite order at $x_0\in\Omega$ if for any sufficiently
small $R>0$, one has
\begin{equation}
\int_{|x-x_0|<R}|u|^2 dx=O(R^M), \ \text{for\ every}\
M\in\mathbb{N}^+.
\end{equation}
\end{defi}
\begin{defi} We say that the operator (4) has SUCP if every solution $\omega$ of the equation
$$\mathscr{H}_{\bf A}^2\omega=\phi\omega,$$ which vanishes of infinite
order at $x_0$ is identically zero in a neighborhood of
$x_0$.\end{defi} So far, the strong unique continuation problem for
second order elliptic operators is well-understood. In the case of
$\Omega=\mathbb{R}^2$, Carleman proved the SUCP of the elliptic
equation with bounded coefficients and $V\in
L_{loc}^\infty(\mathbb{R}^2)$
\begin{equation}-\Delta u=W\cdot\nabla u+Vu\end{equation} by introducing a
weighted $L^2$-estimate, the so-called Carleman estimate \cite{TC}.
For the space dimension $N\geq 3$ with bounded coefficients, N.
Aronszajn, A. Krzywicki and J. Szarski proved the SUCP by means of
Carleman type inequalities, namely, observability inequalities.
Afterwards, D. Jerison, C. E. Kenig, C. D. Sogge treated the
equation (6) with singular potentials $V\in
L^{{N}/{2}}_{loc}(\mathbb{R}^N)$ and $W\in L^\infty(\mathbb{R}^N)$,
$N\geq 3$, by the approach of $L^p-L^q$ Carleman estimate involving
sharp exponents\cite{Kenig,Koch,Sogge}. Afterwards, N. Garofalo and
F. H. Lin gave a new proof for the SUCP of the elliptic operator
$-\Delta u=Vu$ with bounded potential by
applying a variational method in \cite{Lin2}.\\

There is a large body of work on SUCP for (6) with real-valued
coefficients. In this paper, we investigate the complex-valued case.
As a matter of fact, the operator $\mathscr{H}_{\bf A}^2$ can be
decomposed into
\begin{equation} \mathscr{H}_{\bf A}^2\omega=-\Delta\omega+i{\bf
A}\cdot\nabla\omega+i\nabla\cdot({\bf A}\omega)+{\bf A}{\bf
A}^T\omega.
\end{equation}
In \cite{KK1,KK2}, K. Kurata proved the SUCP for (4) with ${\bf
A}{\bf A}^T\in \mathscr{K}^{loc}_N(\Omega)$, where
$\mathscr{K}^{loc}_N(\Omega)$ denotes the Kato class. When the
potential ${\bf A}\in (L^\infty(\Omega))^N$, in effect, it does not
belong to the Kato class. As a result, we can not deduce
corresponding results directly from K. Kurata's work. In this
manuscript, we intend to provide a new approach of SUCP for (4) with
complex-valued coefficients by developing new multipliers. At the
moment one is ready to state the main results.
\begin{thm}
For $N\geq2$, let the complex-valued $\omega\in H^2(\mathbb{B}_1)$
be a solution of the problem \begin{equation}-\Delta\omega+i{\bf
A}\cdot\nabla\omega+i\nabla\cdot({\bf A}\omega)+{\bf A}{\bf
A}^T\omega=\phi(x)\omega \ \ \text{in}\ \mathbb{B}_1,\end{equation}
where $\mathbb{B}_1$ is a unit ball
$\mathbb{B}_1\subset\overline{\Omega}$, ${\bf A}\in
C^1(\overline{\Omega})$ is a real-valued potential vector and the
complex-valued function $\phi\in L^\infty(\mathbb{R}^N)$. If
$\omega$ vanishes of infinite order at $x_0\in\mathbb{B}_1$, then
$\omega\equiv0$ in $\mathbb{B}_1$.
\end{thm}
By virtue of Theorem 1.3, one is able to prove the following
statement for a mixed boundary value problem which is of great
importance in the discussion of exact controllability through
boundary control \cite{LU1}.
\begin{coro} Assume that $\Omega$ is a bounded,
open and connected domain in $\mathbb{R}^N$ with the boundary
$\Gamma\in C^2$, ${\bf A}\in C^1(\overline{\Omega})$ is a
real-valued potential vector and the complex-valued function
$\phi\in L^\infty(\mathbb{R}^N)$. Let $\omega\in H^2(\Omega)$ be the
solution of the mixed boundary problem
$$-\Delta\omega+i{\bf A}\cdot\nabla\omega+i\nabla\cdot({\bf
A}\omega)+{\bf A}{\bf A}^T\omega=\phi(x)\omega\ \ \ \text{in}\
\Omega,$$
$$\omega={\partial\omega}/{\partial\nu}=0\ \text{on}\ \Gamma.
$$
Then $\omega$ is identically 0 in $\Omega$.
\end{coro}
\begin{rem} Theorem 1.3 demonstrates, the asymptotic behavior of the solution
$\omega$ at an interior point $x_0$ determines the interior value of
$\omega$ in $\mathbb{B}_1$. In contrast with Theorem 1.3, Corollary
1.4 indicates, the behavior of solution $\omega$ on the boundary
determines the interior value of $\omega$ in $\Omega$.
\end{rem}
The rest of the paper is organized as follows. First and foremost,
in Section 2, we introduce some useful quantities and their
particular properties. Next, we give an important comparison lemma
and a frequency function. By carefully estimating the derivative of
the frequency function, we reach the conclusion in the final
analysis. In Section 3, as an important application in exact
controllability, we prove Corollary 1.4 in detail.
\section{Proof of the main theorem: A new multiplier method}
First, we introduce several quantities which will serve as useful
tools for our purposes. For every $r\in(0,1)$, we define the
following two quantities
\begin{equation}\Phi(r)\triangleq\int_{\partial\mathbb{B}_r}|\omega|^2dS_x,\end{equation}
where $\mathbb{B}_r$ is centered at the origin with radius $r$,
$\partial\mathbb{B}_r$ denotes its sphere, $dS_x$ stands for the
($N$-1)-dimensional Hausdorff measure on the sphere
$\partial\mathbb{B}_r$.
\begin{equation}\Psi(r)\triangleq\int_{\mathbb{B}_r}(|\mathscr{H}_{\bf
A}\omega|^2-\phi^R|\omega|^2)dV_x,\end{equation} where $\phi^R$
denotes the real part of $\phi$. Actually, we have
\begin{lem} By virtue of divergence theorem, the
following identity holds,
\begin{equation}-{\rm Re}\int_{\partial\mathbb{B}_r}\Big(\nabla|\omega|^2-i{\bf
A}|\omega|^2\Big)\cdot
x/rdS_x=\int_{\mathbb{B}_r}\Big(-2|\mathscr{H}_{\bf
A}\omega|^2+2\phi^R|\omega|^2\Big)dV_x.\end{equation}
\end{lem}
\begin{proof}
On the one hand, \beq\begin{array}{lll}
\int_{\mathbb{B}_r}\mathscr{H}_{\bf A}^2|\omega|^2dV_x&=&-\int_{\partial\mathbb{B}_r}\Big(\nabla|\omega|^2-i{\bf A}\cdot|\omega|^2\Big)\cdot x/rdS_x+\int_{\mathbb{B}_r}{\bf A}\cdot\mathscr{H}_{\bf A}|\omega|^2dV_x\\
\\
&=&-\int_{\partial\mathbb{B}_r}\Big(\nabla|\omega|^2-i{\bf A}|\omega|^2\Big)\cdot x/rdS_x\\
\\
&&+\int_{\mathbb{B}_r}{\bf
A}\cdot(i\overline{\omega}\nabla\omega+i\omega\nabla\overline{\omega}+{\bf
A}|\omega|^2)dV_x.
\end{array}\eneq
On the other hand, \beq\begin{array}{lll}
\int_{\mathbb{B}_r}\mathscr{H}_{\bf A}^2|\omega|^2dV_x&=&\int_{\mathbb{B}_r}\Big(-\overline{\omega}\Delta\omega-\omega\Delta\overline{\omega}-2|\nabla\omega|^2+i\nabla\cdot{\bf A}|\omega|^2\Big)dV_x\\
\\
&&+\int_{\mathbb{B}_r}\Big(i\overline{\omega}{\bf
A}\cdot\nabla\omega+i\omega{\bf A}\cdot\nabla\overline{\omega}+{\bf
A}{\bf A}^T|\omega|^2\Big)dV_x.
\end{array}\eneq
Since $$\nabla\cdot{\bf
A}|\omega|^2=\overline{\omega}\nabla\cdot{\bf A}\omega+\omega{\bf
A}\cdot\nabla\overline{\omega}=\overline{\overline{\omega}\nabla\cdot{\bf
A}\omega+\omega{\bf A}\cdot\nabla\overline{\omega}},$$ then by
combining (12) and (13), we have $$\begin{array}{lll}
&&-{\rm Re}\int_{\partial\mathbb{B}_r}\Big(\nabla|\omega|^2-i{\bf A}|\omega|^2\Big)\cdot x/rdS_x\\
\\
&=&\int_{\mathbb{B}_r}\Big(-\overline{\omega}\Delta\omega-\omega\Delta\overline{\omega}-2|\nabla\omega|^2\Big)dV_x\\
\\
&=&\int_{\mathbb{B}_r}\Big(-2|\nabla\omega|^2+2i\omega{\bf A}\cdot\nabla\overline{\omega}-2i\overline{\omega}{\bf A}\cdot\nabla\omega-2{\bf A}{\bf A}^T|\omega|^2\Big)dV_x\\
\\
&&+\int_{\mathbb{B}_r}\Big(-\overline{\omega}\Delta\omega+i\overline{\omega}{\bf A}\cdot\nabla\omega+i\overline{\omega}\nabla\cdot{\bf A}\omega+{\bf A}{\bf A}^T|\omega|^2\Big)dV_x\\
\\
&&+\int_{\mathbb{B}_r}\Big(-\omega\Delta\overline{\omega}-i\omega{\bf A}\cdot\nabla\overline{\omega}-i\omega\nabla\cdot{\bf A}\overline{\omega}+{\bf A}{\bf A}^T|\omega|^2\Big)dV_x\\
\\
&=&\int_{\mathbb{B}_r}\Big(-2|\mathscr{H}_{\bf
A}\omega|^2+2\phi^R|\omega|^2\Big)dV_x.
\end{array} $$
\end{proof}
Next we calculate the derivatives of $\Phi(r)$ and $\Psi(r)$ with
respect to $r$.
\begin{lem}  The derivatives of $\Phi(r)$ and $\Psi(r)$ with respect to $r$
are presented as follows,
\begin{equation}
\Phi'(r)=(N-1)\Phi(r)/r+2\Psi(r).
\end{equation}
\begin{equation}\begin{array}{lll}\Psi'(r)
&=&(N-2)\Psi(r)/r+(N-2)/r\int_{\mathbb{B}_r}\phi^R|\omega|^2dV_x+2/r\ {\rm Re}\int_{\mathbb{B}_r}{\bf A}\omega\cdot\overline{\mathscr{H}_{\bf A}\omega}dV_x\\
\\
&&+2/r\ {\rm
Re}\int_{\mathbb{B}_r}(x\cdot\nabla\omega)\cdot\overline{\phi\omega}dV_x+2/r\
{\rm Re}\int_{\mathbb{B}_r}\omega x(D{\bf
A})^T\overline{\mathscr{H}_{\bf
A}\omega}^TdV_x\\
\\
&&+2\int_{\partial\mathbb{B}_r}|\nu\cdot(i\nabla\omega+{\bf A}\omega)|^2dS_x-2\ {\rm Re}\int_{\partial\mathbb{B}_r}({\bf A}\omega\cdot\nu)(\overline{\mathscr{H}_{\bf A}\omega\cdot\nu})dS_x\\
\\
&&-\int_{\partial\mathbb{B}_r}\phi^R|\omega|^2dS_x.\end{array}\end{equation}
\end{lem}
\begin{rem}
(14) shows that $\Phi(r)$ and $\Psi(r)$ are closed related with each
other. This relation is very important for our discussion.
\end{rem}
\begin{proof} First we consider the derivative of $\Phi(r)$ with respect to $r$. Indeed, we have
$$\begin{array}{lll}
\Phi'(r)&=&\int_{\partial\mathbb{B}_1}(|\omega(ry)|^2r^{N-1})'_rdS_y\\
\\
&=&\int_{\partial\mathbb{B}_1}\Big((\nabla\omega\cdot y)\overline{\omega}+\omega(\nabla\overline{\omega}\cdot y)r^{N-1}+|\omega|^2(N-1)r^{N-2}\Big)dS_y\\
\\
&=&\int_{\partial\mathbb{B}_r}\Big((\nabla\omega\cdot x/r)\overline{\omega}+\omega(\nabla\overline{\omega}\cdot x/r)\Big)dS_x+(N-1)/r\int_{\partial\mathbb{B}_r}|\omega|^2dS_x\\
\\
&=&(N-1)/r\Phi(r)+{\rm Re}\int_{\partial\mathbb{B}_r}\Big(\nabla|\omega|^2-i{\bf A}|\omega|^2\Big)\cdot x/rdS_x\\
\\
&=&(N-1)\Phi(r)/r+2\Psi(r).
\end{array}$$
As for the derivative of $\Psi(r)$ with respect to $r$, via the
divergence theorem, we have \beq\begin{array}{lll}
\Psi'(r)&=&\int_{\partial\mathbb{B}_r}|\mathscr{H}_{\bf A}\omega|^2dS_x-\int_{\partial\mathbb{B}_r}\phi^R|\omega|^2dS_x\\
\\
&=&1/r\int_{\partial\mathbb{B}_r}|\mathscr{H}_{\bf A}\omega|^2x\cdot x/rdS_x-\int_{\partial\mathbb{B}_r}\phi^R|\omega|^2dS_x\\
\\
&=&1/r\int_{\mathbb{B}_r}{\rm div}(|\mathscr{H}_{\bf A}\omega|^2x)dV_x-\int_{\partial\mathbb{B}_r}\phi^R|\omega|^2dS_x\\
\\
&=&N/r\int_{\mathbb{B}_r}|\mathscr{H}_{\bf
A}\omega|^2dV_x+\underbrace{1/r\int_{\mathbb{B}_r}x\cdot\nabla|\mathscr{H}_{\bf
A}\omega|^2dV_x}_{(I)}-\int_{\partial\mathbb{B}_r}\phi^R|\omega|^2dS_x.\end{array}\eneq
Now we treat the term (I) carefully.
$$\begin{array}{lll}
(I)&=&\displaystyle\sum_{j,k}1/r\int_{\mathbb{B}_r}x_j\nabla_j\Big((i\nabla_k\omega+a_k\omega)\overline{(i\nabla_k\omega+a_k\omega)}\Big)dV_x\\
\\
&=&\displaystyle\sum_{j,k}2/r\ {\rm Re}\int_{\mathbb{B}_r}x_j\Big(\nabla_j(i\nabla_k\omega+a_k\omega)\overline{(i\nabla_k\omega+a_k\omega)}\Big)dV_x\\
\\
&=&\displaystyle\sum_{j,k}2/r\ {\rm Re}\int_{\mathbb{B}_r}ix_j\nabla_j\nabla_k\omega\overline{(i\nabla_k\omega+a_k\omega)}dV_x+\displaystyle\sum_{j,k}2/r\ {\rm Re}\int_{\mathbb{B}_r}x_j\nabla_j(a_k\omega))\overline{(i\nabla_k\omega+a_k\omega)}dV_x\\
\\
&=&\displaystyle\sum_{j,k}2/r\ {\rm Re}\int_{\partial\mathbb{B}_r}ix_j\nabla_j\omega\overline{(i\nabla_k\omega+a_k\omega)}\nu_kdS_x-\displaystyle\sum_{j,k}2/r\ {\rm Re}\int_{\mathbb{B}_r}i\nabla_kx_j\nabla_j\omega\overline{(i\nabla_k\omega+a_k\omega)}dV_x\\
\\
&&-\displaystyle\sum_{j,k}2/r\ {\rm Re}\int_{\mathbb{B}_r}ix_j\nabla_j\omega\overline{\nabla_k(i\nabla_k\omega+a_k\omega)}dV_x+\displaystyle\sum_{j,k}2/r\ {\rm Re}\int_{\mathbb{B}_r}x_j\omega\nabla_ja_k\overline{(i\nabla_k\omega+a_k\omega)}dV_x\\
\\
&&+\displaystyle\sum_{j,k}2/r\ {\rm
Re}\int_{\mathbb{B}_r}a_kx_j\nabla_j\omega\overline{(i\nabla_k\omega+a_k\omega)}dV_x\\
\\
&=&\displaystyle\sum_{j,k}2/r\ {\rm Re}\int_{\partial\mathbb{B}_r}x_j(i\nabla_j\omega+a_j\omega)\overline{(i\nabla_k\omega+a_k\omega)}\nu_kdS_x\\
\\
&&-\displaystyle\sum_{j,k}2/r\ {\rm Re}\int_{\partial\mathbb{B}_r}x_ja_j\omega\overline{(i\nabla_k\omega+a_k\omega)}\nu_kdS_x-\displaystyle\sum_{j,k}2/r\ {\rm Re}\int_{\mathbb{B}_r}i\nabla_kx_j\nabla_j\omega\overline{(i\nabla_k\omega+a_k\omega)}dV_x\\
\\
&&-\displaystyle\sum_{j,k}2/r\ {\rm Re}\int_{\mathbb{B}_r}ix_j\nabla_j\omega\overline{\nabla_k(i\nabla_k\omega+a_k\omega)}dV_x+\displaystyle\sum_{j,k}2/r\ {\rm Re}\int_{\mathbb{B}_r}x_j\omega\nabla_ja_k\overline{(i\nabla_k\omega+a_k\omega)}dV_x\\
\\
&&+\displaystyle\sum_{j,k}2/r\ {\rm Re}\int_{\mathbb{B}_r}a_kx_j\nabla_j\omega\overline{(i\nabla_k\omega+a_k\omega)}dV_x\\
\\
&=&2\int_{\partial\mathbb{B}_r}|\nu\cdot\mathscr{H}_{\bf A}\omega|^2dS_x-2\ {\rm Re}\int_{\partial\mathbb{B}_r}({\bf A}\omega\cdot\nu)(\overline{\mathscr{H}_{\bf A}\omega\cdot\nu})dS_x\\
\\
&&-2/r\int_{\mathbb{B}_r}|\mathscr{H}_{\bf A}\omega|^2dV_x+{2}/{r}\ {\rm Re}\int_{\mathbb{B}_r}{\bf A}\omega\cdot\overline{\mathscr{H}_{\bf A}\omega}dV_x+2/r\ {\rm Re}\int_{\mathbb{B}_r}(x\cdot\nabla\omega)\cdot\overline{\mathscr{H}_{\bf A}^2\omega}dV_x\\
\\
&&+2/r\ {\rm Re}\int_{\mathbb{B}_r}\omega x(D{\bf
A})^T\overline{\mathscr{H}_{\bf A}\omega}^TdV_x\\
\\
&=&2\int_{\partial\mathbb{B}_r}|\nu\cdot\mathscr{H}_{\bf A}\omega|^2dS_x-2\ {\rm Re}\int_{\partial\mathbb{B}_r}({\bf A}\omega\cdot\nu)(\overline{\mathscr{H}_{\bf A}\omega\cdot\nu})dS_x\\
\\
&&-2/r\int_{\mathbb{B}_r}|\mathscr{H}_{\bf A}\omega|^2dV_x+2/r\ {\rm Re}\int_{\mathbb{B}_r}{\bf A}\omega\cdot\overline{\mathscr{H}_{\bf A}\omega}dV_x+2/r\ {\rm Re}\int_{\mathbb{B}_r}(x\cdot\nabla\omega)\cdot\overline{\phi\omega}dV_x\\
\\
&&+2/r\ {\rm Re}\int_{\mathbb{B}_r}\omega x(D{\bf
A})^T\overline{\mathscr{H}_{\bf A}\omega}^TdV_x.\\
\\
\end{array}$$
%where $$\Theta_{\bf A}\triangleq\left(\begin{array}{cccc}\nabla_{1}a_1&\nabla_{1}a_2,&\cdots&\nabla_{1}a_N\\
%\nabla_{2}a_1&\nabla_{2}a_2&\cdots&\nabla_{2}a_N\\
%\vdots&\vdots&\cdots&\vdots\\
%\nabla_{N}a_1&\nabla_{N}a_2&\cdots&\nabla_{N}a_N\\
%\end{array}\right).$$
Finally, keeping in mind the definition of $\Psi(r)$, we reach the
conclusion immediately.
\end{proof}
Next we show an important comparison lemma.
\begin{lem}
There exists an $r_0\in(0,1)$ such that for every $r\in(0,r_0)$, we
have \begin{equation} \int_{\mathbb{B}_r}|\omega|^2dV_x\leq
r\int_{\partial\mathbb{B}_r}|\omega|^2dS_x.\end{equation}
\end{lem}
\begin{proof}
On the one hand, $$
\begin{array}{lll}&&\int_{\mathbb{B}_r}\mathscr{H}_{\bf
A}^2|\omega|^2\cdot(r^2-|x|^2)dV_x\\
\\
&=&\int_{\mathbb{B}_r}|\omega|^2\cdot\overline{\mathscr{H}_{\bf
A}^2(r^2-|x|^2)}dV_x+\int_{\partial\mathbb{B}_r}|\omega|^2\cdot\overline{{\partial(r^2-|x|^2)}/{\partial\nu_{i\mathscr{H}_{\bf
A}}}}dS_x\\
\\
&=&\int_{\mathbb{B}_r}|\omega|^2\Big(2N+2i{\bf A}\cdot
x-i\nabla\cdot{\bf A}(r^2-|x|^2)+{\bf A}{\bf A}^T(r^2-|x|^2)\Big)dV_x\\
\\
&&-2r\int_{\partial\mathbb{B}_r}|\omega|^2dS_x.
\end{array}$$
On the other hand,
$$\begin{array}{lll}&&\int_{\mathbb{B}_r}\mathscr{H}_{\bf
A}^2|\omega|^2\cdot(r^2-|x|^2)dV_x\\
\\
&=&\int_{\mathbb{B}_r}\Big(-\overline{\omega}\Delta\omega-\omega\Delta\overline{\omega}-2|\nabla\omega|^2\Big)(r^2-|x|^2)dV_x+\int_{\mathbb{B}_r}{\bf
A}{\bf
A}^T|\omega|^2(r^2-|x|^2)dV_x\\
\\
&&+\int_{\mathbb{B}_r}\Big(i\nabla\cdot{\bf
A}|\omega|^2+i\overline{\omega}{\bf A}\cdot\nabla\omega+i\omega{\bf
A}\cdot\nabla\overline{\omega}\Big)(r^2-|x|^2)dV_x\\
\\
&=&\int_{\mathbb{B}_r}\Big(-2|\mathscr{H}_{\bf
A}\omega|^2+2\phi^R|\omega|^2\Big)(r^2-|x|^2)dV_x+\int_{\mathbb{B}_r}{\bf
A}{\bf
A}^T|\omega|^2(r^2-|x|^2)dV_x\\
\\
&&+\int_{\mathbb{B}_r}\Big(i\nabla\cdot{\bf
A}|\omega|^2+i\overline{\omega}{\bf A}\cdot\nabla\omega+i\omega{\bf
A}\cdot\nabla\overline{\omega}\Big)(r^2-|x|^2)dV_x.
\end{array}$$
As a result,
$$\int_{\mathbb{B}_r}\Big(2N|\omega|^2+2|\mathscr{H}_{\bf
A}\omega|^2(r^2-|x|^2)-2\phi^R|\omega|^2(r^2-|x|^2)\Big)dV_x=2r\int_{\partial\mathbb{B}_r}|\omega|^2dS_x.$$
When $\|\phi^R\|_{L^\infty}>0$, then we choose $r_0\in(0,{1}/{2})$
such that
$$r_0^2\leq{(N-1)}/{\|\phi^R\|_{L^\infty}}.$$ It follows
immediately that
$$\int_{\mathbb{B}_r}|\omega|^2dV_x\leq\int_{\mathbb{B}_r}\Big(N|\omega|^2-\phi^R|\omega|^2(r^2-|x|^2)\Big)dV_x\leq
r\int_{\partial\mathbb{B}_r}|\omega|^2dS_x.$$ When
$\|\phi^R\|_{L^\infty}=0$, then it is evident $$
\int_{\mathbb{B}_r}|\omega|^2dV_x\leq
{r}/{N}\int_{\partial\mathbb{B}_r}|\omega|^2dS_x.$$
\end{proof}
Assume that there exists a small $r_1\in(0,1)$ such that
\begin{equation}\Phi(r)\neq0\ \ \text{for}\ \ \forall\ r\in(0,r_1).\end{equation} Define
the frequency function
\begin{equation}\digamma(r)\triangleq{r\Psi(r)}/{\Phi(r)},\ \ \
r\in(0,r_1).\end{equation} Let $r^\ast\triangleq\min\{r_0,r_1\}$,
and we set
\begin{equation}\beth_{r^\ast}\triangleq\Big\{r\in(0,r^\ast):\digamma(r)>1\Big\}.\end{equation}
With the above definitions, we have the following inequality for the
frequency function.
\begin{lem}
Under the assumptions (18)-(20), there exists a positive constant
$\tau=\tau(N,\phi)$ which is independent of $r$ such that
$\digamma'(r)$ is estimated in a uniform fashion,
$$\digamma'(r)\geq-\digamma(r)\tau.$$
\end{lem}
\begin{proof} Actually, from (17)-(20), we have
$$\int_{\mathbb{B}_r}|\mathscr{H}_{\bf
A}\omega|^2dV_x>({1}/{r^2}-\|\phi^R\|_{L^\infty})\int_{\mathbb{B}_r}|\omega|^2dV_x.$$
Indeed, $$\begin{array}{lll}
\int_{\mathbb{B}_r}|\mathscr{H}_{\bf A}\omega|^2dV_x&=&\int_{\mathbb{B}_r}(|\mathscr{H}_{\bf A}\omega|^2-\phi^R|\omega|^2)dV_x+\int_{\mathbb{B}_r}\phi^R|\omega|^2dV_x\\
\\
&>&1/r\int_{\partial\mathbb{B}_r}|\omega|^2dS_x+\int_{\mathbb{B}_r}\phi^R|\omega|^2dV_x\\
\\
&\geq&{1}/{r^2}\int_{\mathbb{B}_r}|\omega|^2dV_x+\int_{\mathbb{B}_r}\phi^R|\omega|^2dV_x\\
\\
&\geq&({1}/{r^2}-\|\phi^R\|_{L^\infty})\int_{\mathbb{B}_r}|\omega|^2dV_x.
\end{array}$$
This indicates the integral $\int_{\mathbb{B}_r}|\mathscr{H}_{\bf
A}\omega|^2dV_x$ is the dominating part in $\Psi(r)$. By calculating
$\digamma'(r)$ with respect to $r$, we have the following identity,
$$\begin{array}{lll}
\digamma'(r)&=&\digamma(r)\Big(\Psi'(r)/\Psi(r)+1/r-\Phi'(r)/\Phi(r)\Big)\\
\\
&=&\digamma(r)\underbrace{\Big({2\int_{\partial\mathbb{B}_r}|\nu\cdot(i\nabla\omega+{\bf
A}\omega)|^2dS_x}/{{\rm Re}
\int_{\partial\mathbb{B}_r}(x/r\cdot\nabla\omega)\overline{\omega}dS_x}}_{(II)}\\
\\
&&-\underbrace{{2\ {\rm
Re}\int_{\partial\mathbb{B}_r}({x}/{r}\cdot\nabla\omega)
\overline{\omega}dS_x}/{\int_{\partial\mathbb{B}_r}|\omega|^2dS_x}\Big)}_{(II')}\\
\\
&&+\digamma(r)\Big\{\underbrace{{(N-2)}/{r}\int_{\mathbb{B}_r}\phi^R|\omega|^2dV_x}_{(III)}+\underbrace{{2}/{r}
\ {\rm Re}\int_{\mathbb{B}_r}{\bf A}\omega\cdot\overline{\mathscr{H}_{\bf A}\omega}dV_x}_{(IV)}\\
\\
&&+\underbrace{{2}/{r}\ {\rm
Re}\int_{\mathbb{B}_r}(x\cdot\nabla\omega)\cdot\overline{\phi\omega}dV_x}_{(V)}+\underbrace{{2}/{r}\
{\rm Re}\int_{\mathbb{B}_r}\omega x(D{\bf
A})^T\overline{\mathscr{H}_{\bf
A}\omega}^TdV_x}_{(VI)}\\
\\
&&\underbrace{-2\ {\rm Re}\int_{\partial\mathbb{B}_r}({\bf A}\omega\cdot\nu)(\overline{\mathscr{H}_{\bf A}\omega\cdot\nu})dS_x}_{(VII)}\\
\\
&&\underbrace{-\int_{\partial\mathbb{B}_r}\phi^R|\omega|^2dS_x}_{(VIII)}
\Big\}\Big/\Big\{{1}/{2}\int_{\partial\mathbb{B}_r}{x}/{r}\cdot\nabla|\omega|^2dS_x\Big\}.
\end{array}$$
We estimate each term respectively. For (II)-(II'), we apply
H\"{o}lder's inequality and obtain
$$\begin{array}{lll}&&(II)-(II')\\
\\
&=& 2\int_{\partial\mathbb{B}_r}|\nu\cdot(i\nabla\omega+{\bf
A}\omega)|^2dS_x/ {\rm
Re}\int_{\partial\mathbb{B}_r}({x}/{r}\cdot\nabla\omega)\overline{\omega}dS_x\\
\\
&&-2\ {\rm
Re}\int_{\partial\mathbb{B}_r}({x}/{r}\cdot\nabla\omega)\overline{\omega}dS_x/\int_{\partial\mathbb{B}_r}|\omega|^2dS_x\\
\\
&=&2\int_{\partial\mathbb{B}_r}|\nu\cdot(i\nabla\omega+{\bf
A}\omega)|^2dS_x/{\rm
Re}\int_{\partial\mathbb{B}_r}({x}/{r}\cdot(\nabla\omega-i{\bf
A}\omega))\overline{\omega}dS_x\\
\\
&&-2\ {\rm
Re}\int_{\partial\mathbb{B}_r}({x}/{r}\cdot(\nabla\omega-i{\bf
A}\omega))\overline{\omega}dS_x/\int_{\partial\mathbb{B}_r}|\omega|^2dS_x\\
\\
&\geq&2\int_{\partial\mathbb{B}_r}|\nu\cdot(i\nabla\omega+{\bf
A}\omega)|^2dS_x/\Big(\sqrt{\int_{\partial\mathbb{B}_r}|({x}/{r}\cdot(\nabla\omega-i{\bf
A}\omega))|^2dS_x}\sqrt{\int_{\partial\mathbb{B}_r}|\omega|^2
dS_x}\Big)\\
\\
&&-2\sqrt{\int_{\partial\mathbb{B}_r}|({x}/{r}\cdot(\nabla\omega-i{\bf
A}\omega))|^2dS_x}\sqrt{\int_{\partial\mathbb{B}_r}|\omega|^2
dS_x}/\int_{\partial\mathbb{B}_r}|\omega|^2dS_x\\
\\
&\geq&0.
\end{array}$$
In addition, we have
\begin{lem} There exists a constant $C^\ast(\phi)$ independent of $r$ such that
$$
{\int_{\partial\mathbb{B}_r}|\nu\cdot(i\nabla\omega+{\bf
A}\omega)|^2dS_x}/{\Psi(r)}\leq {C^\ast(\phi)}/{r}.$$
\end{lem}
\begin{proof} Indeed, by multiplying ${\bf
H}(x)\cdot\mathscr{H}_{\bf A}\omega$ to $\mathscr{H}_{\bf
A}^2\omega=\phi\omega$ and integrating by parts, we have the
following identity,
$$\begin{array}{lll}
&&-{1}/{2}\int_{\partial\mathbb{B}_r}\Big|{\partial
\omega}/{\partial\nu_{i\mathscr{H}_{\bf A}}}\Big|^2\cdot\Big({\bf
H}(x)\cdot\nu\Big)dS_x \\
\\
&=&{1}/{2}\int_{\mathbb{B}_r}\Big(\nabla\cdot{\bf H(x)}\Big)\cdot\Big|\mathscr{H}_{\bf A}\omega\Big|^2dV_x-{\rm Im}\int_{\mathbb{B}_r}\phi\omega\cdot\Big({\bf H}(x)\cdot\overline{\mathscr{H}_{\bf A}\omega}\Big)dV_x\\
\\
&&-{\rm Re}\int_{\mathbb{B}_r}\mathscr{H}_{\bf A}\omega(D{\bf
H})^T\overline{\mathscr{H}_{\bf A}^T\omega}dV_x-{\rm
Re}\int_{\mathbb{B}_r} \overline{\omega}\mathscr{H}_{\bf
A}\omega\Xi_{\bf A}{\bf H}^TdV_x.
\end{array}
$$
Since
$$
\int_{\mathbb{B}_r}|\mathscr{H}_{\bf
A}\omega|^2dV_x\geq({1}/{r^2}-\|\phi^R\|_{L^\infty})\int_{\mathbb{B}_r}|\omega|^2dV_x,
$$
by choosing $${\bf H}(x)\triangleq{x}/{r},$$ we have the following
estimate, $$\begin{array}{lll}
&&{1}/{2}\int_{\partial\mathbb{B}_r}\Big|{\partial
\omega}/{\partial\nu_{i\mathscr{H}_{\bf A}}}\Big|^2dS_x\\
\\
&\leq&{N}/{(2r)}\int_{\mathbb{B}_r}|\mathscr{H}_{\bf
A}\omega|^2dV_x+1/2{\|\phi\|_{L^\infty}}\int_{\mathbb{B}_r}(|\omega|^2+
|\mathscr{H}_{\bf A}\omega|^2)dV_x\\
\\
&&+{1}/{r}\int_{\mathbb{B}_r}|\mathscr{H}_{\bf
A}\omega|^2dV_x+{1}/{2}\max\|\Xi_{\bf
A}\|_F\int_{\mathbb{B}_r}(|\omega|^2+
|\mathscr{H}_{\bf A}\omega|^2)dV_x\\
\\
&=&\underbrace{\Big({(N+2)}/{(2r)}+{1/2(\|\phi\|_{L^\infty}+\max\|\Xi_{\bf A}\|_F)}\Big)}_{\alpha}\int_{\mathbb{B}_r}|\mathscr{H}_{\bf A}\omega|^2dV_x\\
\\
&&+\underbrace{{1/2\Big(\|\phi\|_{L^\infty}+\max\|\Xi_{\bf A}\|_F\Big)}}_{\beta}\int_{\mathbb{B}_r}|\omega|^2dV_x\\
\\
&\leq&\Big(\alpha+\beta{r^2}/{(1-
r^2\|\phi^R\|_{L^\infty})}\Big)\int_{\mathbb{B}_r}|\mathscr{H}_{\bf
A}\omega|^2dV_x.
\end{array}
$$
where $\|\Xi_{\bf A}\|_F$ denotes the Frobenius norm of the test
matrix $\Xi_{\bf A}$. Since
$$\Psi(r)=\int_{\mathbb{B}_r}(|\mathscr{H}_{\bf A}\omega|^2-\phi^R|\omega|^2)dV_x\geq{(1-2r^2\|\phi^R\|_{L^\infty})}/{(1-r^2\|\phi^R\|_{L^\infty})}\int_{\mathbb{B}_r}|\mathscr{H}_{\bf A}\omega|^2dV_x,$$
therefore, $$\begin{array}{lll}
&&{\int_{\partial\mathbb{B}_r}|\nu\cdot(i\nabla\omega+{\bf
A}\omega)|^2dS_x}/{\Psi(r)}\\
\\
&\leq&
{(1-r^2\|\phi^R\|_{L^\infty})}/{(1-2r^2\|\phi^R\|_{L^\infty})}\Big(2\alpha+{2r^2}/{(1-
r^2\|\phi^R\|_{L^\infty})}\beta\Big).\end{array}$$ The conclusion
follows immediately.
\end{proof}
Taking Lemma 2.6 into account and noticing the fact $\digamma(r)>1$,
for the term (III), we have
$$\begin{array}{lll}
&&(III)/\Psi(r)\\
\\
&=&{(N-2)\int_{\mathbb{B}_r}\phi^R|\omega|^2dV_x}/{(r\Psi(r))}\\
\\
&\leq&{(N-2)\|\phi^R\|_{L^\infty}\Phi(r)}/{\Psi(r)}\\
\\
&\leq& r(N-2)\|\phi^R\|_{L^\infty}.\end{array}
$$
For the term (IV),
$$\begin{array}{lll}&&|(IV)/\Psi(r)|\\
\\
&=&|2\ {\rm Re}\int_{\mathbb{B}_r}{\bf
A}\omega\cdot\overline{\mathscr{H}_{\bf
A}\omega}dV_x|/(r\Psi(r))\\
\\
&\leq&2\int_{\mathbb{B}_r}|{\bf
A}\omega|\cdot|\overline{\mathscr{H}_{\bf
A}\omega}|dV_x/(r\Psi(r))\\
\\
&\leq&\Big({1}/{(2\epsilon r)}\int_{\mathbb{B}_r}|{\bf
A}\omega|^2dV_x+{2\epsilon}/{r}\int_{\mathbb{B}_r}|\mathscr{H}_{\bf
A}\omega|^2dV_x\Big)/\Psi(r)\\
\\
&=&\Big({1}{(2\epsilon r)}\int_{\mathbb{B}_r}{\bf A}{\bf
A}^T|\omega|^2dV_x+{2\epsilon}/{r}\int_{\mathbb{B}_r}(|\mathscr{H}_{\bf
A}\omega|^2-\phi^R|\omega|^2)dV_x+{2\epsilon}/{r}\int_{\mathbb{B}_r}\phi^R|\omega|^2dV_x\Big)/\Psi(r)\\
\\
&\leq&\|{\bf A}{\bf
A}^T\|_{L^\infty}\Phi(r)/(2\epsilon\Psi(r))+2\epsilon/r+2\epsilon\|\phi^R\|_{L^\infty}\Phi(r)/\Psi(r).
\end{array}$$
Let $\epsilon={r}/{2},$ since $\digamma(r)>1$, then $$
|(IV)/\Psi(r)|\leq\|{\bf A}{\bf
A}^T\|_{L^\infty}+1+r^2\|\phi^R\|_{L^\infty}.$$ For the term $(V)$,
$$\begin{array}{lll}&&|(V)/\Psi(r)|\\
\\
&=&|2\ {\rm
Re}\int_{\mathbb{B}_r}(x\cdot\nabla\omega)\cdot\overline{\phi\omega}dV_x|/(r\Psi(r))\\
\\
&=&|2\ {\rm
Re}\int_{\mathbb{B}_r}(x\cdot(\nabla\omega-i{\bf A}\omega))\cdot\overline{\phi\omega}dV_x|/(r\Psi(r))\\
\\
&\leq&\Big({\|\phi\|_{L^\infty}}/{(2\epsilon
r)}\int_{\mathbb{B}_r}|\omega|^2dV_x+{2/r\epsilon\|\phi\|_{L^\infty}}\int_{\mathbb{B}_r}|x\cdot\mathscr{H}_{\bf
A}\omega|^2dV_x\Big)/\Psi(r)\\
\\
&\leq&\Big({\|\phi\|_{L^\infty}}/{(2\epsilon
r)}\int_{\mathbb{B}_r}|\omega|^2dV_x+2r\epsilon\|\phi\|_{L^\infty}\int_{\mathbb{B}_r}|\mathscr{H}_{\bf
A}\omega|^2dV_x\Big)/\Psi(r)\\
\\
&=&\Big({\|\phi\|_{L^\infty}}/{(2\epsilon
r)}\int_{\mathbb{B}_r}|\omega|^2dV_x+2r\epsilon\|\phi\|_{L^\infty}\int_{\mathbb{B}_r}(|\mathscr{H}_{\bf
A}\omega|^2-\phi^R|\omega|^2)dV_x\\
\\&&+2r\epsilon\|\phi\|_{L^\infty}\int_{\mathbb{B}_r}\phi^R|\omega|^2dV_x\Big)/\Psi(r)\\
\\
&\leq&r\|\phi\|_{L^\infty}
/(2\epsilon)+2r\epsilon\|\phi\|_{L^\infty}+2r^3\epsilon\|\phi\|_{L^\infty}^2.
\end{array}$$
Let $\epsilon={r}/{2}$, then $$
|(V)/\Psi(r)|\leq\|\phi\|_{L^\infty}(1+r^2+r^4\|\phi\|_{L^\infty}).$$
For the term (VI),
$$\begin{array}{lll}&&|(VI)/\Psi(r)|\\
\\
&=&|2\ {\rm Re}\int_{\mathbb{B}_r}\omega x(D{\bf
A})^T\overline{\mathscr{H}_{\bf
A}\omega}^TdV_x|/(r\Psi(r))\\
\\
&\leq&2\int_{\mathbb{B}_r}\|\omega x\|_2\|(D{\bf
A})^T\overline{\mathscr{H}_{\bf
A}\omega}^T\|_2dV_x/(r\Psi(r))\\
\\
&\leq&2\int_{\mathbb{B}_r}\|\omega x\|_2\|(D{\bf
A})^T\|_{F}\|\mathscr{H}_{\bf
A}\omega\|_2dV_x/(r\Psi(r))\\
\\
&\leq&2\max\|(D{\bf
A})^T\|_{F}\int_{\mathbb{B}_r}|\omega||\mathscr{H}_{\bf
A}\omega|dV_x/\Psi(r)\\
\\
&=&\Big({\max\|(D{\bf A})^T\|_{F}}/{(2\epsilon)
}\int_{\mathbb{B}_r}|\omega|^2dV_x+2\epsilon\max\|(D{\bf
A})^T\|_{F}\Big\{\int_{\mathbb{B}_r}(|\mathscr{H}_{\bf
A}\omega|^2-\phi^R|\omega|^2)dV_x\\
\\
&&+\int_{\mathbb{B}_r}\phi^R|\omega|^2dV_x\Big\}\Big)/\Psi(r)\\
\\
&\leq&\max\|(D{\bf A})^T\|_{F}(r^2/(2\epsilon)+2\epsilon+2\epsilon
r^2\|\phi^R\|_{L^\infty}).
\end{array}$$
Let $\epsilon={r}/{2}$, then
$$|(VI)/\Psi(r)|\leq\max\|(D{\bf
A})^T\|_{F}(2r+r^3\|\phi^R\|_{L^\infty}),$$ where $\|(D{\bf
A})^T\|_{F}$ denotes the Frobenius norm of $(D{\bf A})^T$.\\

Let $\epsilon={r}/{2}$. By Schwartz's inequality, we have the
estimate below for the term (VII),
$$\begin{array}{lll}&&|(VII)/\Psi(r)|\\
\\
&=&|-2\ {\rm Re}\int_{\partial\mathbb{B}_r}({\bf
A}\omega\cdot\nu)(\overline{\mathscr{H}_{\bf
A}\omega\cdot\nu})dS_x|/\Psi(r)\\
\\
&\leq&2\int_{\partial\mathbb{B}_r}|{\bf
A}\omega\cdot\nu||\overline{\mathscr{H}_{\bf
A}\omega\cdot\nu}|dS_x/\Psi(r)\\
\\
&\leq&\Big({1}/{(2\epsilon)}\int_{\partial\mathbb{B}_r}|{\bf
A}\omega|^2dS_x+2\epsilon\int_{\partial\mathbb{B}_r}|\mathscr{H}_{\bf
A}\omega\cdot\nu|^2dS_x\Big)/\Psi(r)\\
\\
&\leq&\|{\bf A}{\bf
A}^T\|_{L^\infty}\int_{\partial\mathbb{B}_r}|\omega|^2dS_x/(2\epsilon\Psi(r))+2\epsilon\int_{\partial\mathbb{B}_r}|\mathscr{H}_{\bf
A}\omega\cdot\nu|^2dS_x/\Psi(r)\\
\\
&\leq&r\|{\bf A}{\bf
A}^T\|_{L^\infty}/(2\epsilon)+2\epsilon\int_{\partial\mathbb{B}_r}|\mathscr{H}_{\bf
A}\omega\cdot\nu|^2dS_x/\Psi(r)\\
\\
&\leq&\|{\bf A}{\bf A}^T\|_{L^\infty}+C^\ast(\phi),
\end{array}$$
where $C^\ast(\phi)$ is from Lemma 2.6.\\

For the last term (VIII), a simple calculation leads to
$$|(VIII)/\Psi(r)|=\int_{\partial\mathbb{B}_r}\phi^R|\omega|^2dS_x/\Psi(r)\leq
r\|\phi^R\|_{L^\infty}.
$$
From the above estimates, we conclude that there exists a positive
constant $\tau=\tau(N,\phi)$ which is independent of $r$ such that
$$\digamma'(r)\geq-\digamma(r)\tau.$$
\end{proof}
It follows that $\exp(\tau r)\digamma(r)$ is monotonously increasing
on $(0,r^\ast)$, that is to say, $$\exp(\tau
r)\digamma(r)\leq\exp(\tau r^\ast)\digamma(r^\ast).$$ Keeping in
mind the case $\digamma\leq 1$, we know that, $\digamma(r)$ is
bounded on $(0,r^\ast)$. Since
$$ \Phi'(r)={(N-1)}/{r}\Phi(r)+2\Psi(r),$$ then $$
\Big(\log({\Phi(r)}/{r^{N-1}})\Big)'={2\Psi(r)}/{\Phi(r)}={2\digamma(r)}/{r}\leq
{C(\tau)}/{r}.$$ We integrate from $\gamma$ to $2\gamma$, then
$$\log({2^{1-N}\Phi(2\gamma)}/{\Phi(\gamma)})\leq C(\tau)\log
2.$$ It follows that $$ \Phi(2\gamma)\leq
2^{C(\tau)+N-1}\Phi(\gamma).$$ Finally, integrating with respect to
$\gamma$ gives
$$ \int_{\mathbb{B}_{2\gamma}}|\omega|^2dV_x\leq
2^{C(\tau)+N-1}\int_{\mathbb{B}_\gamma}|\omega|^2dV_x.$$ Since
$\mathbb{B}_1$ is connected, then our theorem follows immediately.
\begin{rem}
It is of great interest to explore the strong unique continuation
for a variety of Schr\"{o}dinger operators with singular or
nonlinear potentials by the multiplier method. More results will be
available in sequential papers.
\end{rem}
\section{Proof of Corollary 1.4} In this section, we show an important application of Theorem 1.3 in
\cite{LU1}.\\
%\begin{coro} Let $\Omega$ be a bounded,
%open and connected domain in $\mathbb{R}^N$ with the boundary
%$\Gamma\in C^2$. Let $\omega\in H^2$ be a solution of
%$$\left\{\begin{array}{lll}
%\mathscr{H}_{\bf A}^2\omega=\phi\omega&\text{in}&\Omega,\\
%\omega=\frac{\partial\omega}{\partial\nu}=0&\text{on}& \Gamma .
%\end{array}\right.$$
%Then $\omega$ is identically 0 in $\Omega$. \end{coro}

{\rm\bf Proof of Corollary 1.4}: Let $\mathbb{B}$ be an arbitrarily
small open ball such that
$$\Gamma\cap\mathbb{B}\neq\varnothing.$$ Set
$$\Omega^1\triangleq\Omega\cup\mathbb{B},$$ and define
$$\omega
^1\triangleq\left\{\begin{array}{lll}
\omega&\text{in}&\Omega;\\
\\
0&\text{in}&\mathbb{B}\backslash\Omega.
\end{array}\right.$$
It is sufficient to verify that $\omega^1\in H^2$. Denote by
$\omega^1_j$, $\omega^1_{jk}$ the extension by zero to $\Omega^1$ of
the derivatives $\nabla_j\omega$, $\nabla_j\nabla_k\omega$,
$j,k=1,\cdots,N$. Then $\omega_j$, $\omega_{jk}\in L^2(\Omega^1)$
and it is necessary to demonstrate that, for
$\forall\zeta\in\mathscr{D}(\Omega^1)$, $$
\int_{\Omega^1}\omega^1\nabla_j\overline{\zeta}
dx=-\int_{\Omega^1}\omega^1_j\overline{\zeta} dx,$$ and $$
\int_{\Omega^1}\omega^1_j\nabla_k\overline{\zeta}
dx=-\int_{\Omega^1}\omega^1_{jk}\overline{\zeta} dx.$$ Indeed, since
$\omega^1_j=\omega_{jk}^1\equiv0$ outside of $\Omega$,
$\zeta\equiv0$ on $\Gamma\backslash(\Gamma\cap\mathbb{B})$ and
$\omega={\partial\omega}/{\partial\nu}\equiv0$ on
$\Gamma\cap\mathbb{B}$, we have $$\begin{array}{lll}
&&\int_{\Omega^1}\omega^1\nabla_j\overline{\zeta}
dx=\int_{\Omega}\omega\nabla_j\overline{\zeta}
dx=\int_{\Gamma}\omega\overline{\zeta}\nu_jd\Gamma-\int_{\Omega}(\nabla_j\omega)
\overline{\zeta}
dx\\
\\
&=&\int_{\Gamma\cap\mathbb{B}}\omega\overline{\zeta}\nu_jd\Gamma-\int_{\Omega}(\nabla_j\omega)
\overline{\zeta} dx=-\int_\Omega(\nabla_j\omega)\overline{\zeta}
dx=-\int_{\Omega^1}\omega^1_j\overline{\zeta} dx,
\end{array}$$
and $$\begin{array}{lll}
&&\int_{\Omega^1}\omega^1_j\nabla_k\overline{\zeta}
dx=\int_{\Omega}\nabla_j\omega\nabla_k\overline{\zeta}
dx=\int_{\Gamma}\nabla_j\omega\overline{\zeta}\nu_kd\Gamma-\int_{\Omega}(\nabla_k\nabla_j\omega)
\overline{\zeta} dx\\
\\
&=&\int_{\Gamma\cap\mathbb{B}}\nabla_j\omega\overline{\zeta}\nu_kd\Gamma-\int_{\Omega}(\nabla_k\nabla_j\omega)
\overline{\zeta}
dx=-\int_\Omega(\nabla_k\nabla_j\omega)\overline{\zeta}
dx=-\int_{\Omega^1}\omega^1_{jk}\overline{\zeta} dx.
\end{array}$$
Thus, the result is concluded due to the connectness of $\Omega$.

\section*{Acknowledgement} This project is partially supported by Natural Science
Foundation of Jiangsu Province (BK 20130598), National Natural
Science Foundation of China (NSFC 71273048, 71473036, 11471072), the
Scientific Research Foundation for the Returned Overseas Chinese
Scholars, Open Research Fund Program of Jiangsu Key Laboratory of
Engineering Mechanics, Southeast University (LEM16B06), Fundamental
Research Funds for the Central Universities on the Field Research of
Commercialization of Marriage between China and Vietnam (No.
2014B15214). In particular, the authors also express their deep
gratitude to the referees for their careful reading and useful
remarks.

\end{document}